\documentclass[letterpaper, 10 pt, conference]{ieeeconf}

\IEEEoverridecommandlockouts                              

\usepackage{cite}
\usepackage{amsmath,amssymb,amsfonts}
\usepackage{algorithmic}
\usepackage{graphicx}
\usepackage{textcomp}
\usepackage{bm}
\usepackage{hyperref}
\usepackage{siunitx}
\DeclareMathAlphabet\mathbfcal{OMS}{cmsy}{b}{n}

\newcommand*{\genbf}[1]{\ifmmode\mathbf{#1}\else\textbf{#1}\fi}
\newcommand{\norm}[1]{\left\lVert#1\right\rVert}

\newcommand{\Phix}[2]{\bm{\Phi}^{\mathbf{x}}[\mathbf{#1},\mathbf{#2}]}
\newcommand{\Phiu}[2]{\bm{\Phi}^{\mathbf{u}}[\mathbf{#1},\mathbf{#2}]}
\newcommand{\Psix}[0]{\bm{\Psi}^{\mathbf{x}}}
\newcommand{\Psiu}[0]{\bm{\Psi}^{\mathbf{u}}}

\newcommand{\Emme}[0]{\mathbfcal{M}}

\newcommand{\xb}[0]{\mathbf{x}}
\newcommand{\ub}[0]{\mathbf{u}}
\newcommand{\Fb}[0]{\mathbf{F}}
\newcommand{\wb}[0]{\mathbf{w}}

\usepackage{tikz}
\usetikzlibrary{arrows}

\makeatletter
\newcommand\@erelb@r[1]{%
  \mathrel{\tikz[baseline=-.5ex]\draw[#1] (0,0)--(0.3,0);}
}
\newcommand{\erelbar}[1]{\@erelbar#1}
\def\@erelbar#1#2{%
  \ifcase\numexpr#1*4+#2\relax
    \@erelb@r{-}\or     
    \@erelb@r{->}\or    
    \@erelb@r{-|}\or    
    \@erelb@r{->|}\or   
    \@erelb@r{<-}\or    
    \@erelb@r{<->}\or   
    \@erelb@r{<-|}\or   
    \@erelb@r{<->}\or   
    \@erelb@r{|-}\or    
    \@erelb@r{|->}\or   
    \@erelb@r{|-|}\or   
    \@erelb@r{|<->|}\or 
    \@erelb@r{|<-}\or   
    \@erelb@r{|<->}\or  
    \@erelb@r{|<-|}\or  
    \@erelb@r{|<->|}    
  \else
    \@wrong
  \fi
}
\makeatother

\newtheorem{theorem}{Theorem}
\newtheorem{corollary}{Corollary}
\newtheorem{remark}{Remark}
\newtheorem{definition}{Definition}

\def\BibTeX{{\rm B\kern-.05em{\sc i\kern-.025em b}\kern-.08em
    T\kern-.1667em\lower.7ex\hbox{E}\kern-.125emX}}
\begin{document}
%
\title{Neural System Level Synthesis:  Learning over All Stabilizing Policies for Nonlinear Systems}
\author{Luca Furieri, Clara Luc\'{i}a Galimberti, and Giancarlo Ferrari-Trecate
    \thanks{Authors are with the Institute of Mechanical Engineering, École Polytechnique Fédérale de Lausanne, Switzerland. E-mails: {\tt\footnotesize \{luca.furieri, clara.galimberti, giancarlo.ferraritrecate\}@epfl.ch}.}
    \thanks{Research supported by the Swiss National Science Foundation under the NCCR Automation (grant agreement 51NF40\textunderscore 80545).} 
}

\maketitle

\begin{abstract}
We address the problem of designing stabilizing control policies for nonlinear systems in discrete-time, while minimizing an arbitrary cost function. When the system is linear and the cost is convex, the System Level Synthesis (SLS) approach offers an effective solution based on convex programming. Beyond this case, a globally optimal solution cannot be found in a tractable way, in general. 
In this paper, we develop a parametrization of \emph{all and only} the control policies stabilizing a given time-varying nonlinear system in terms of the combined effect of 1) a strongly stabilizing base controller and 2) a stable SLS operator to be freely designed. Based on this result, we  propose a Neural SLS (Neur-SLS) approach guaranteeing closed-loop stability during and after parameter optimization, without requiring any constraints to be satisfied. We exploit recent Deep Neural Network (DNN) models based on Recurrent Equilibrium Networks (RENs) to learn over a rich class of nonlinear stable operators, and demonstrate the effectiveness of the proposed approach in numerical examples.

\end{abstract}

\section{Introduction}
\label{sec:introduction}

The optimal control of nonlinear systems subject to global asymptotic stability guarantees is one of the most important and challenging problems in control theory. Despite a rich body of research in nonlinear control methods \cite{sastry2013nonlinear}, the general Nonlinear Optimal Control (NOC) problem is well understood only when the system dynamics are linear and the cost is convex. Classical approaches for dealing with continuous-time and discrete-time NOC include dynamic programming and the associated Bellman optimality principle, and the maximum principle for deterministic problems \cite{bertsekas2011dynamic, doya2007bayesian}. Unfortunately, application of these methods to general NOC problems is hindered by methodological and computational challenges \cite{doya2007bayesian}.

An alternative approach is provided by receding-horizon control schemes such as Nonlinear Model Predictive Control (NMPC) \cite{rawlings2017model}. The NMPC strategy is based on real-time optimization; a finite-horizon open-loop control problem is solved at each time instant for defining the control input at the same time. One of the main drawbacks of NMPC is that the resulting  control policy is not available offline in explicit form. Furthermore, several applications may not allow for sufficient computational resources to solve a mathematical program in real-time.

More recently, Reinforcement Learning (RL) and Deep Neural Networks (DNN) have emerged as powerful tools helping agents learn how to make sense of and optimally interact with complex environments  \cite{sutton2018reinforcement}. Despite the success of DNN control policies for finite-horizon control tasks \cite{onken2021neural,gama2021graph}, RL techniques cannot provide closed-loop stability guarantees in general. As a result, their applicability is limited  to systems that are not safety-critical.

Several methods to learn provably stabilizing DNN control policies have been proposed. We divide them into two categories. First, \emph{constrained optimization approaches} \cite{berkenkamp2018safe,gu2021recurrent,pauli2021offset} ensure global or local stability by enforcing appropriate Lyapunov-like inequalities during optimization. However, optimizing over conservative stability constraints may excessively limit the set of admissible policies.
Further, enforcing constraints such as Linear Matrix Inequalities (LMIs) becomes a computational bottleneck in large-scale applications. Second, \emph{unconstrained optimization approaches} seek to characterize classes of control policies with \emph{built-in} stability guarantees \cite{wang2021learning,furieri2021distributed}. The unconstrained approach allows one to  learn over real-valued parameters through unconstrained gradient descent, without compromising stability for arbitrarily chosen parameter values and whilst enjoying a low computational burden.  The main challenge is to characterize stabilizing policy classes that are broad enough to achieve optimized performance. An original perspective on parametrizing stabilizing control policies is based on the nonlinear System Level Synthesis (SLS) \cite{ho2020system}, which proposes to directly learn over stable nonlinear closed-loop maps — rather than over stabilizing nonlinear control policies. Unfortunately, nonlinear functional equality constraints must be satisfied during optimization \cite{ho2020system}. This aspect poses obstacles to practical implementation beyond the well-known linear case  \cite{wang2019system}.


\subsection{Contributions}
Motivated by the SLS framework, we establish an unconstrained parametrization of \emph{all} stabilizing nonlinear control policies for nonlinear time-varying systems in discrete-time. Our key finding is that the intractable constraints of nonlinear SLS \cite{ho2020system} need not be imposed during optimization, if a stabilizing base control policy is known, and if this base control policy defines itself a stable map from its inputs to its outputs. In this case, all and only the other stabilizing policies are parametrized in terms of a freely chosen stable SLS operator. 
 A main practical benefit is that, no matter how suboptimally this stable SLS operator may be chosen, closed-loop stability is guaranteed by design. Our result extends nonlinear variants of the Youla parametrization \cite{anantharam1984stabilization,fujimoto1998state} to discrete-time system. Further, we recover the recent result of \cite{wang2021learning} as a subcase, where the authors address the NOC problem for linear systems with nonlinear costs. Last, we propose using deep RENs \cite{kim2018standard,revay2021recurrent} as a finite-dimensional approximation of nonlinear stable operators, thus extending the well-known idea of Finite Impulse Responses (FIRs) approximations \cite{wang2019system} which is valid for linear SLS only.  
 The resulting Neur-SLS optimal control method is tested on formation control problems involving vehicles with nonlinear dynamics and nonlinear costs with collision avoidance terms. 

\smallskip













\smallskip

\textbf{Notation:}
The set of all sequences $\mathbf{x} = (x_0,x_1,x_2,\ldots)$, where $x_t \in \mathbb{R}^n$ for all $t\in \mathbb{N}$ and $\mathbb{N}$ is the set of natural numbers, is denoted as $\ell^n$. Moreover,  $\mathbf{x}$ belongs to $\ell_p^n \subset \ell^n$ with $p \in \mathbb{N}$ if $\norm{\mathbf{x}}_p = \left(\sum_{t=0}^\infty |x_t|^p\right)^{\frac{1}{p}} < \infty$, where $|\cdot|$ denotes any vector norm. We say that $\xb \in \ell^n_\infty$ if $\operatorname{sup}_{t}|x_t|< \infty$.
We use the notation $x_{j:i}$ to refer to the truncation of $\mathbf{x}$ to the finite-dimensional vector $(x_i,x_{i+1},\ldots,x_{j})$. An operator $\mathbf{A}:\ell^n \rightarrow \ell^m$ is said to be \emph{causal} if $\mathbf{A}(\mathbf{x}) = (A_0(x_0),A_1(x_{1:0}),\ldots,A_t(x_{t:0}),\ldots)$. If in addition $A_t(x_{t:0}) = A_t(0,x_{t-1:0})$, then $\mathbf{A}$ is said to be strictly causal. Similarly, $A_{j:i}(x_{j:0}) =  (A_i(x_{i:0}),A_{i+1}(x_{i+1:0}),\ldots,A_j(x_{j:0}))$. For a matrix $M \in \mathbb{R}^{m \times n}$, $M\mathbf{x} = (Mx_0,Mx_1,\ldots) \in \ell^m$. 
An operator $\mathbf{A}:\ell^n\rightarrow \ell^m$ is said to be $\ell_p$-stable if it is causal and $\mathbf{A}(\mathbf{a}) \in \ell_{p}^m$ for all $\mathbf{a} \in \ell_{p}^n$. Equivalently, we  write $\mathbf{A} \in \mathcal{L}_{p}$. 

\section{Problem Statement and Nonlinear System Level Synthesis}
\label{sec:problem_statement}

We consider nonlinear discrete-time time-varying systems
\begin{equation}
    \label{eq:system}
    x_{t} = f_t(x_{t-1:0},u_{t-1:0})+w_t\,,~~~t= 1,2,\ldots\,,
\end{equation}
where $x_t \in \mathbb{R}^n$ denotes the state vector, $u_t \in \mathbb{R}^m$ denotes the control input and $w_t \in \mathbb{R}^n$ stands for unknown process noise with $ w_0 =x_0 $. In \emph{operator form}, system \eqref{eq:system} is equivalent to
\begin{equation}
\label{eq:operator_form}
    \mathbf{x} = \mathbf{F}(\mathbf{x},\mathbf{u}) + \mathbf{w}\,,
\end{equation}
where $\mathbf{F}:\ell^n\times \ell^m \rightarrow \ell^n$ is the strictly causal operator such that $\mathbf{F}(\mathbf{x},\mathbf{u}) = (0,f_1(x_0,u_0),\ldots,f_t(x_{t-1:0},u_{t-1:0}),\ldots)$. Further, we assume that $w_t \sim \mathcal{D}_t(\mu_t,\Sigma_t)$, i.e., $w_t$ is distributed according to an unknown distribution $\mathcal{D}_t$ with expected value $\mu_t$ and covariance matrix $\Sigma_t\succ 0$, and that $\mathbf{w}$ belongs to $\ell_{p}^n$. 

In order to control the behavior of system \eqref{eq:system} we consider nonlinear, dynamic, time-varying feedback control policies
\begin{equation}
    \label{eq:control}
    \mathbf{u} = \mathbf{K}(\mathbf{x}) = (K_0(x_0),K_1(x_{1:0}),\ldots, K_t(x_{t:0}),\ldots)\,,
\end{equation}
where $\mathbf{K}$ is a causal operator $\mathbf{K}:\ell^n\ \rightarrow \ell^m$ to be designed.  It is easy to see that for each sequence of disturbances $\mathbf{w} \in \ell^n_p$ the closed-loop system \eqref{eq:system}-\eqref{eq:control} produces unique trajectories. Hence, the closed-loop mappings $\mathbf{w} \erelbar{21}(\mathbf{x},\mathbf{u})$ are well-defined. Specifically, for a system $\mathbf{F}$ and a controller $\mathbf{K}$, we denote the corresponding induced closed-loop maps $\mathbf{w} \erelbar{21}\mathbf{x}$ and $\mathbf{w} \erelbar{21}\mathbf{u}$ as $\Phix{F}{K}$ and $\Phiu{F}{K}$, respectively. Therefore, we have $\mathbf{x} = \Phix{F}{K}(\mathbf{w})$ and $\mathbf{u} = \Phiu{F}{K}(\mathbf{w})$ for all $\mathbf{w} \in \ell_{p}^n$.

Our goal is to synthesize a control policy $\mathbf{K}$ that complies with two requirements:
\begin{itemize}
    \item[(R1)] The closed-loop maps $\Phix{F}{K}$ and $\Phiu{F}{K}$ are $\ell_p$-stable.
    \item[(R2)] A loss function
    \begin{equation}
        \label{eq:cost}
         J = \mathbb{E}_{w_{T:0}}\left[\sum_{t=0}^{T}l(x_t,u_t)\right]\,,
    \end{equation}
     is globally minimized, where $l$ is piecewise differentiable and $l(x,u)\geq0$ for all $(x,u) \in \mathbb{R}^{m+n}$.
\end{itemize}
Requirement  (R1) is a \emph{hard} constraint to be satisfied for all the policies we optimize over. In other words, we require fail-safe learning, in the sense that closed-loop stability must be guaranteed during and after optimization of the parameters describing the policies. 
Instead — as is standard in neural network training problems — we treat (R2) as a ``soft'' requirement. 
We do not expect that a gradient-descent-based algorithm finds the globally optimal solution for any  initialization — this is generally impossible for problems beyond Linear Quadratic Gaussian (LQG) control, which enjoy convexity of the cost and linearity of the optimal policies \cite{tang2021analysis,furieri2020first}. 
Last, we note that the theoretical results of Section~\ref{sec:main} also apply to convergent infinite-horizon costs  \cite{bertsekas2011dynamic}; we focus on the finite-horizon cost \eqref{eq:cost} to enable the DNN implementation of Section~\ref{sec:Neur-SLS}.   


We are ready to formulate the NOC problem as
\begin{alignat}{3}
\operatorname{NOC:}~~~&\min_{\mathbf{K}(\cdot)}&& \qquad \mathbb{E}_{w_{T:0}}\left[\sum_{t=0}^{T}l(x_t,u_t)\right]\label{eq:LOSSNOC}\\
&\operatorname{s.t.}~~ && x_{t} = f_t(x_{t-1:0},u_{t-1:0})+w_t\,, ~~w_0 = x_0\,,\nonumber\\
&~~&&u_t = K_t(x_{t:0})\,,~~\forall t =0,1,\ldots\,,\nonumber\\
&~~&&(\Phix{F}{K},\Phiu{F}{K}) \in \mathcal{L}_{p}\,.\label{eq:CONSTRNOC}
\end{alignat}

\smallskip

\noindent  The NOC problem can also be approached as a \emph{learning problem}: based on many observed noisy state/input trajectories — which represent our dataset — iteratively improve the policy $\mathbf{K}(\cdot)$ by descending the gradient of the empiric evaluation of the loss function \eqref{eq:LOSSNOC}. However, learning over the space of stabilizing control policies leads to the intractable constraints \eqref{eq:CONSTRNOC} in general. The idea behind the SLS approach \cite{wang2019system, ho2020system} is to circumvent the difficulty of characterizing stabilizing controllers, by  directly designing stable closed-loop maps. Define the set of \emph{achievable} closed-loop maps for $\mathbf{F}$ as
\begin{equation}
    \label{eq:Phi_ach}
    \mathcal{CL}[\mathbf{F}] = \{(\Phix{F}{K},\Phiu{F}{K})|~\mathbf{K}\text{ is causal}\}\,,
\end{equation}
and the set of all \emph{achievable  and stable} closed-loop maps as
\begin{equation}
    \label{eq:Phi_ach_stable}
    \mathcal{CL}_p[\Fb] = \{(\Psix,\Psiu) \in \mathcal{CL}[\Fb]|~(\Psix,\Psiu) \in \mathcal{L}_{p}\}\,.
\end{equation}
Note that, if $(\Psix,\Psiu) \in \mathcal{CL}_p[\mathbf{F}]$, then $\mathbf{x} = \Psix(\mathbf{w}) \in \ell_{p}^n$ and $\mathbf{u} = \Psiu(\mathbf{w})\in \ell_{p}^m$ for all $\wb \in \ell_{p}^n$. Based on Theorem III.3 of \cite{ho2020system}, and adding the requirement that the closed-loop maps must belong to $\mathcal{L}_p$, we summarize the main SLS result for nonlinear discrete-time systems.

\begin{theorem}[\emph{Nonlinear SLS parametrization \cite{ho2020system}}]
\label{th:SLS}
~
\begin{enumerate}
    \item The set $\mathcal{CL}_p[\Fb]$ of all achievable and stable closed-loop responses admits the following characterization:
    \begin{align}
        \mathcal{CL}_p[\Fb] = \{&(\Psix,\Psiu)|~~(\Psix,\Psiu)\text{ are causal}\,,\label{SLS_equality0}\\
        &\Psix = \mathbf{F}(\Psix,\Psiu)+\mathbf{I}\,,\label{SLS_equality1}\\
        &(\Psix,\Psiu) \in \mathcal{L}_{p} \label{SLS_equality2} \}\,.
    \end{align}
    

\item For any $(\Psix,\Psiu) \in \mathcal{CL}_p[\Fb]$, the causal controller 
\begin{equation}
   \mathbf{u} = \mathbf{K}(\mathbf{x}) = \Psiu((\Psix)^{-1}(\mathbf{x}))\label{eq:equivalent_representation}\,,
\end{equation}
 is unique, and it produces the stable closed-loop responses $(\Psix,\Psiu)$. 
 
\end{enumerate}
\end{theorem}

\smallskip

Theorem~\ref{th:SLS} clarifies that any policy $\mathbf{K}(\xb)$ achieving  $\ell_p$-stable closed-loop maps can be described in terms of two operators $(\Psix,\Psiu) \in \mathcal{L}_{p}$ complying with the nonlinear functional equality \eqref{SLS_equality1}. Therefore, the NOC problem admits an equivalent Nonlinear SLS (N-SLS) formulation:
\begin{alignat}{3}
\operatorname{N-SLS:}~~~&\min_{(\Psix,\Psiu)}&& \qquad \mathbb{E}_{w_{T:0}}\left[\sum_{t=0}^{T}l(x_t,u_t)\right] \label{eq:NSLS}\\
&\operatorname{s.t.}~~ && x_t = \Psi^x_t(w_{t:0})\,,~~~u_t = \Psi^u_t(w_{t:0})\,,\nonumber\\
&~~&&(\Psix,\Psiu)\in \mathcal{CL}_p[\Fb]\,, \forall t = 0,1,\ldots  \nonumber
\end{alignat}

\noindent According to Theorem~\ref{th:SLS}, the constraint $(\Psix,\Psiu)\in \mathcal{CL}_p[\Fb]$ is equivalent to requiring that $(\Psix,\Psiu)$ 
comply with both \eqref{SLS_equality1} and \eqref{SLS_equality2}, which are challenging to satisfy simultaneously. 
The constraint \eqref{SLS_equality1} simply defines the operator $\Psix$ in terms of $\Psiu$ and can be computed explicitly because $\mathbf{F}$ is strictly causal. Instead, it is hard to select $\Psiu \in \mathcal{L}_{p}$ such that the corresponding $\Psix$ satisfies $\Psix \in \mathcal{L}_{p}$. 
The paper \cite{ho2020system} suggests directly searching over $\ell_p$-stable operators $(\Psix,\Psiu)$ and abandoning the goal of complying with \eqref{SLS_equality1} exactly.  One can then study robust stability when \eqref{SLS_equality1} only holds approximately. However, this way of proceeding may result in conservative control policies. In this work, we turn our attention to the problem of complying with \eqref{SLS_equality1}-\eqref{SLS_equality2} \emph{by design}, thus leading to an unconstrained parametrization of all stabilizing controllers.

\section{Main Result: Unconstrained Parametrization of all Stabilizing Controllers}
\label{sec:main}

In this section we show that, if \emph{a single} stabilizing control policy $\mathbf{K}'(\mathbf{x})$ with $\mathbf{K}' \in \mathcal{L}_{p}$ is available for the nonlinear system $\mathbf{F}$, it is possible to parametrize \emph{all other} stabilizing control policies in terms of  stable maps $\Emme \in \mathcal{L}_{p}$ by applying the control input 
\begin{equation}
\label{eq:stabilizing_input}
    \mathbf{u} = \mathbf{K}'(\mathbf{x})+\Emme(\mathbf{w})\,,
\end{equation}
where $\mathbf{w}$ is reconstructed as $\mathbf{x}-\mathbf{F}(\mathbf{x},\mathbf{u})$.
Furthermore, if $\mathbf{K}'$ is stabilizing, but not itself an operator in $\mathcal{L}_p$, the control policy \eqref{eq:stabilizing_input} achieves $\ell_p$-stable closed-loop maps $(\Psix,\Psiu) \in \mathcal{CL}_p[\Fb]$  for any $\Emme \in \mathcal{L}_{p}$. 
From now on, we consider the system
\begin{align}
    \mathbf{x} &= \mathbf{F}(\mathbf{x},\mathbf{u})+\mathbf{w}\,,\quad \mathbf{u} = \mathbf{K}'(\mathbf{x}) + \mathbf{v}\,,\label{eq:CL_1}\\
    \mathbf{v} &= \Emme(\mathbf{w})\,,\label{eq:CL_3}
\end{align}
where $\mathbf{K}'$ is the \emph{base controller}, and the operator $\Emme:\ell^n\rightarrow\ell^m$ must be designed in order to comply with the closed-loop stability requirement (R1) and the  optimality requirement (R2). The combined effect of $\mathbf{K}'(\mathbf{x})$ and $\Emme(\mathbf{w})$ defines an overall control policy $\mathbf{K}$ such that $\mathbf{u}(\xb) = \mathbf{K}(\mathbf{x})$ is equivalent to \eqref{eq:stabilizing_input}. 

Akin to the Youla parametrization for linear systems \cite{zhou1998essentials}, the role of a base controller $\mathbf{K}'$ is to appropriately stabilize the system; this allows defining a set of ``stable coordinates'' and then freely optimize over $\Emme$.  
\begin{definition}[$\ell_p$-Input-State stabilizing controller]
\label{def:stabilizing}
Given $p \in \mathbb{N}\cup \infty$, we say that the base controller $\mathbf{K}'$ is $\ell_p$-Input-State ($\ell_p$-IS) stabilizing  for $\mathbf{F}$ if the map $(\mathbf{w},\mathbf{v}) \erelbar{21} (\mathbf{x},\mathbf{u})$ defined by \eqref{eq:CL_1} lies in $\mathcal{L}_p$.
\end{definition} 

\smallskip


We remark that the notion of $\ell_p$-IS stabilizability is linked to those of incremental passivity \cite{koelewijn2021incremental} and input-state-stability (ISS) \cite{jiang2001input} for discrete-time systems. For instance, if a system is ISS-stable, then the control policy $\mathbf{K'} = 0$ is $\ell_\infty$-IS stabilizing.  While this paper does not deal with the computation of a base controller, modern design methods include NMPC \cite{rawlings2017model} and incremental dissipativity \cite{koelewijn2021incremental}.  Last, note that our results may be extended to \emph{local} notions of stability, for which a locally stabilizing base controller would be sufficient \cite{fujimoto1998state}.

\smallskip

\begin{definition}[Strongly $\ell_p$-IS stabilizing  controller]
\label{def:strongly_stabilizing}
We say that the base controller $\mathbf{K}'$ is \emph{strongly $\ell_p$-IS stabilizing} if it is  $\ell_p$-IS stabilizing, and additionally $\mathbf{K}'\in \mathcal{L}_{p}$.
\end{definition}

 \smallskip
 

We are  ready to state our main result.


\begin{theorem}
\label{th:result} 
Consider the closed-loop system \eqref{eq:CL_1}-\eqref{eq:CL_3}. 
\begin{enumerate}
    \item Assume that $\mathbf{K}'$ is $\ell_p$-IS stabilizing. Then,  $(\Phix{\Fb}{\mathbf{K}},\Phiu{\Fb}{\mathbf{K}}) \in \mathcal{CL}_p[\Fb]$ for every $\Emme \in \mathcal{L}_{p}$.
    \item If, in addition, $\mathbf{K}'$ is strongly  $\ell_p$-IS stabilizing, then, for any $(\Psix,\Psiu) \in \mathcal{CL}_p[\Fb]$, there exists $\Emme \in \mathcal{L}_{p}$ such that the control policy \eqref{eq:stabilizing_input}
    achieves the closed-loop maps $(\Phix{\Fb}{\mathbf{K}},\Phiu{\Fb}{\mathbf{K}}) = (\Psix,\Psiu)$.
\end{enumerate}
\end{theorem}


\begin{proof}
    $1)$  is a restatement of Definition~\ref{def:stabilizing} with the substitution $\mathbf{v} = \Emme(\wb)$. We prove $2)$.   Let $(\Psix,\Psiu) \in \mathcal{CL}_p[\Fb]$ and choose
    \begin{equation}
        \label{eq:choice_Emme}
        \Emme = -\mathbf{K}'(\Psix)+\Psiu\,.
    \end{equation}
    Note that the composition of operators in $\mathcal{L}_p$ remains in $\mathcal{L}_p$.
    Then, since $(\Psix,\Psiu) \in \mathcal{L}_p$,  we conclude that $\Emme \in \mathcal{L}_p$. It remains to prove that \eqref{eq:choice_Emme} is such that the resulting control policy \eqref{eq:stabilizing_input} achieves the closed-loop maps $(\Phix{\Fb}{\mathbf{K}},\Phiu{\Fb}{\mathbf{K}}) = (\Psix,\Psiu)$. We prove this fact by induction. For the inductive step, we assume that, for any $j \in \mathbb{N}$, we have $\Phi^x_i = \Psi^x_i$ and $\Phi^u_i = \Psi^u_i$ for all $i \in \mathbb{N}$ with $0\leq i \leq j$, where we have dropped the notation $[\mathbf{F},\mathbf{K}]$ in the interest of readability. Since $(\Phix{F}{K},\Phiu{F}{K})$ are closed-loop maps induced by $\mathbf{K}$ and $(\Psix,\Psiu) \in \mathcal{CL}_p[\Fb]$, then
    \begin{equation}
        \label{eq:achievability_proof}
        \Phi^x_{j\hspace{-0.04cm}+\hspace{-0.04cm}1} \hspace{-0.1cm}=\hspace{-0.1cm} F_{j\hspace{-0.04cm}+\hspace{-0.04cm}1}(\Phi^x_{j:0},\Phi^u_{j:0})\hspace{-0.04cm}+\hspace{-0.01cm}I,\Psi^x_{j\hspace{-0.01cm}+\hspace{-0.04cm}1} \hspace{-0.1cm}=\hspace{-0.1cm} F_{j\hspace{-0.04cm}+\hspace{-0.01cm}1}(\Psi^x_{j:0},\Psi^u_{j:0})\hspace{-0.01cm}+\hspace{-0.04cm}I.
    \end{equation}
    Then, by  \eqref{eq:stabilizing_input},
    \eqref{eq:choice_Emme}, \eqref{eq:achievability_proof}, and $x_{j+1:0} = \Phi^x_{j+1:0}(w_{j+1:0})$,
    \begin{align}
        \Phi_{j+1}^u = &K'_{j+1}\left(F_{j+1:0}(\Phi^x_{j:0},\Phi^u_{j:0})+I\right) -\nonumber\\
        &-K'_{j+1}\left(F_{j+1:0}(\Psi^x_{j:0},\Psi^u_{j:0})+I\right)+\Psi^u_{j+1}\,.\label{eq:passage_inductive}
    \end{align}
    By inductive assumption  $\Phi^x_i = \Psi^x_i$ and $\Phi^u_i = \Psi^u_i$ for every $0\leq i\leq j$. Hence, \eqref{eq:passage_inductive} simplifies to $\Phi^u_{j+1} = \Psi^u_{j+1}$. By \eqref{eq:achievability_proof},  $\Phi^x_{j+1}=\Psi^x_{j+1}$. 
    For the base case  $j=0$, by strict causality of $\mathbf{F}$ and achievability, we have $\Phi_0^x = \Psi_0^x = I$.  Then, \eqref{eq:stabilizing_input} and \eqref{eq:choice_Emme} imply $\Phi^u_0 = K'_0(\Phi^x_0)-K'_0(\Psi^x_0)+\Psi^u_0 = \Psi^u_0$.
\end{proof}

To summarize, Theorem~\ref{th:result} establishes that the nonlinear functional equalities involved in the N-SLS problem \eqref{eq:NSLS} can be dropped by tackling the problem in two steps. First, obtain an $\ell_p$-IS stabilizing, yet suboptimal control policy $\mathbf{u}= \mathbf{K}'(\mathbf{x})$. Second, explore the space of operators $\Emme \in \mathcal{L}_{p}$ such that $\mathbf{u}= \mathbf{K}'(\mathbf{x})+\Emme(\mathbf{x}-\mathbf{F}(\xb,\ub))$ optimizes the closed-loop behavior. The corresponding closed-loop maps will lie in $\mathcal{L}_p$ by design. If $\mathbf{K}'$ is also strongly $\ell_p$-IS stabilizing, then, the proposed parametrization is \emph{complete}; all the achievable closed-loop maps — including the ones that are globally optimal for the N-SLS problem \eqref{eq:NSLS} —  are achieved by appropriately selecting $\Emme\in \mathcal{L}_{p}$ in \eqref{eq:stabilizing_input}.  We illustrate the proposed control architecture in Figure~\ref{fig:scheme}.
\begin{figure}
  \centering
  \vspace{2pt}
  \includegraphics[width=0.99\linewidth]{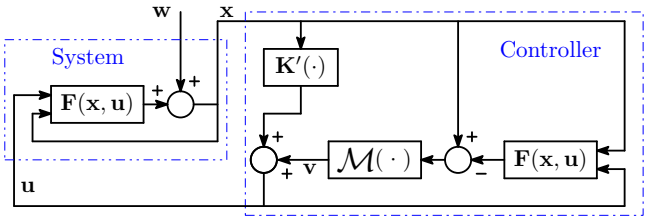}
  \caption{Proposed parametrization of all stabilizing controllers in terms of one freely chosen operator $\Emme \in \mathcal{L}_p$.}
  \label{fig:scheme}
\end{figure}


\begin{remark}
The assumption of strong stabilizability is pivotal to parametrizing all stabilizing controllers. The classical work \cite{anantharam1984stabilization} has parametrized classes of nonlinear stabilizing controllers in continuous-time based on  strongly stabilizing base controllers. The case of linear systems is no exception; a linear strongly stabilizing controller yields a doubly-coprime factorization of the plant, which enables the well-known Youla parametrization \cite{zhou1998essentials} of linear stabilizing controllers. From a theoretical point of view, the results of Theorem~\ref{th:result} can  be interpreted as the counterpart  of \cite{anantharam1984stabilization,fujimoto1998state}   for discrete-time control systems, using the lens of SLS.  
\end{remark}

\subsection{The case of linear systems with nonlinear costs}
Consider a linear system of the form
\begin{equation}
\label{eq:linear_system}
    x_t = Ax_{t-1}+Bu_{t-1}+w_t\,,
\end{equation}
and let $z$ denote the time-shift operator. The classical Youla  parametrization \cite{zhou1998essentials} states that all \emph{linear} stabilizing control policies $\mathbf{u} = \mathbf{K}\xb$ can be written in terms of a stable Youla parameter $\mathbf{Q}\in \mathcal{TF}_s$, where $\mathcal{TF}_s$ denotes the set of stable transfer matrices. For the state-feedback case, the Youla approach \cite{zhou1998essentials} simplifies to finding  a $K'$ such that $(A+BK')$ is Schur-stable, and then expressing all other linear stabilizing control policies as
\begin{equation}
\label{eq:Youla_Linear}
    \mathbf{u} = \frac{\left(K'-\mathbf{Q}\right)}{z}(A\mathbf{x}+B\ub) + \mathbf{Q}\mathbf{x} \,,\quad \mathbf{Q} \in \mathcal{TF}_s\,.
\end{equation}
The class of linear control policies is globally optimal for standard LQG problems, and it allows optimizing over $\mathbf{Q} \in \mathcal{TF}_s$ using  convex-programming. However, these properties are lost  when the controller is distributed \cite{furieri2020sparsity}, or the cost function is nonconvex — even if the system dynamics are linear. For the latter scenario, the recent paper \cite{wang2021learning} characterizes all stabilizing (and contractive) nonlinear policies for output-feedback linear systems by using a nonlinear and contractive Youla parameter. We provide an alternative proof of such result with adaptation to $\mathcal{L}_p$ closed-loop maps in the state-feedback case as an immediate corollary of Theorem~\ref{th:result}. 
\begin{corollary}
 Consider the linear system \eqref{eq:linear_system} and assume that $(A,B)$ is stabilizable. Let $K'$ be such that $(A+BK')$ is Schur-stable. Then, all control policies achieving nonlinear closed-loop maps in $\mathcal{CL}_p$ are expressed, for $\Emme \in \mathcal{L}_{p}$, as
 \begin{equation}
     \label{eq:all_policies}
     \mathbf{u} = K'\mathbf{x} + \Emme\left(\mathbf{x}-z^{-1}\left(A\mathbf{x}+B\mathbf{u}\right)\right)\,.
 \end{equation}
 
 \end{corollary}
 \begin{proof}
     The linear base controller $K'$ is $\ell_p$-IS stabilizing by assumption and it is strongly stabilizing since $K'\mathbf{x} \in \ell^n_{p}$ for each $\xb \in \ell_{p}^n$. 
     The conclusion follows by applying point 2 of Theorem~\ref{th:result}.
 \end{proof}
Note that, as expected, the linear Youla parametrization \eqref{eq:Youla_Linear} is a special case of the proposed parametrization \eqref{eq:all_policies} with $\Emme = \mathbf{Q}-K'$ and $\mathbf{Q} \in \mathcal{TF}_s$.

\section{Neural System Level Synthesis}
\label{sec:Neur-SLS}
In this section, we propose the Neur-SLS approach for tackling the NOC problem. Our goal is to exploit Theorem~\ref{th:result} in order to characterize highly expressive DNN control policies which are stabilizing by design — irrespective of how ``badly'' one may choose the DNN  weights. Then, appropriate DNN training leads to optimized performance.

\subsection{Finite-dimensional approximations of $\mathcal{L}_p$ using RENs}
An obstacle to directly applying the results of Theorem~\ref{th:result} in practice is that the space $\mathcal{L}_{p}$ is infinite-dimensional. Hence, a finite-dimensional approximation of $\Emme \in \mathcal{L}_{p}$ is necessary for implementation. When linear systems are considered, the SLS approach  \cite{wang2019system}  suggests searching over FIR transfer matrices $\mathbf{M} = \sum_{i=0}^N M[i]z^{-i} \in \mathcal{TF}_s$. 
One can then optimize over the finitely many real matrices $M[i]$ for every $i=0,\ldots,N$ and obtain near-optimal solutions by selecting the FIR order $N$ to be large enough.

However, the FIR approach  limits the search to linear control policies. Recently, \cite{kim2018standard,revay2021recurrent} have proposed finite-dimensional  DNN approximations of nonlinear $\mathcal{L}_{p}$ operators. 
An operator $\Emme:\ell^n\rightarrow \ell^m$ is a REN if the relationship $\mathbf{u} = \Emme(\wb)$ is generated by the following dynamical system:
\begin{equation}
\label{eq:RENs}
        \begin{bmatrix}
    \xi_t \\  v_t \\ u_t
  \end{bmatrix}
  =
  \overbrace{
  \begin{bmatrix}
    A_{1} & B_1 & B_2 \\
    C_1 & D_{11} & D_{12} \\
    C_2 & D_{21} & D_{22}
  \end{bmatrix}
  }^{W}
  \begin{bmatrix}
    \xi_{t-1} \\
    \sigma(v_t)\\
    w_t
  \end{bmatrix}\,,\quad \xi_{-1} = 0\,,
\end{equation}
where $\xi_t \in \mathbb{R}^{q}$, $v_t \in \mathbb{R}^r$, and $\sigma:\mathbb{R} \rightarrow \mathbb{R}$ — the activation function — is applied element-wise. Further, $\sigma(\cdot)$ must be piecewise differentiable and with first derivatives restricted to the interval $[0,1]$. As noted in \cite{revay2021recurrent}, RENs subsume many existing DNN architectures. In general, RENs define \textit{deep equilibrium network} models \cite{bai2019deep} due to the implicit relationships between the signals involved in \eqref{eq:RENs}. By restricting $D_{11}$ to be strictly lower-triangular, the outputs of \eqref{eq:RENs} can be computed explicitly, thus significantly speeding-up computations \cite{revay2021recurrent}. 
Further, note that the nonlinear relationship defined by \eqref{eq:RENs} is ``arbitrarily deep''; the nonlinearity $\sigma(\cdot)$ is recursively applied on the REN internal state $\xi_{t-1}$ and the input $w_t$ for $r$ times, where $r$ is the chosen dimension of $v_t$. 

For an arbitrary choice of $W$, the map $\Emme$ induced by \eqref{eq:RENs} may not lie in $\mathcal{L}_p$. The work \cite{revay2021recurrent} provides an explicit smooth mapping $\Theta:\mathbb{R}^d \rightarrow \mathbb{R}^{(q+r+m) \times (q+r+n)}$ from unconstrained training parameters $\theta \in \mathbb{R}^d$ to  a matrix $W=\Theta(\theta) \in \mathbb{R}^{(q+r+m) \times (q+r+n)}$ defining \eqref{eq:RENs}, with the property that the corresponding operator $\Emme[\theta]$ lies in $ \mathcal{L}_2$ by design.\footnote{Furthermore, RENs enjoy contractivity — although the theoretical results of this paper do not rely on this property.} 

\subsection{Neur-SLS}
By combining the theoretical results of Theorem~\ref{th:result} with REN approximations of operators in $\mathcal{L}_2$, we can cast the following unconstrained Neur-SLS learning problem:
\begin{alignat}{3}
&~&&\quad~~~~~~~~~~ \operatorname{\textbf{Neur-SLS}}\nonumber\\
   \quad &\min_{\theta \in \mathbb{R}^d}&& \frac{1}{S}\sum_{s=1}^S\left[\sum_{t=0}^{T}l(x_t^s,u_t^s)\right]\label{eq:loss}\\
  &\operatorname{s.t.}~~ && x_{t}^s = f_t(x_{t-1:0}^s,u_{t-1:0}^s)+w_t^s\,, ~~w_0^s = x_0^s\,,\label{eq:dynamics_s}\\
  &\begin{bmatrix}\xi_t^s\\v_t^s\\u_t^s\end{bmatrix} \hspace{-0.15cm}&&=\hspace{-0.15cm} \begin{bmatrix}0\\0\\K'_t(x_{t:0}^s)\end{bmatrix}\hspace{-0.1cm}+ \hspace{-0.1cm}\Theta(\theta) \begin{bmatrix}\xi_{t-1}^s\\\sigma(v_t^s)\\x_{t}^s\hspace{-0.1cm}-\hspace{-0.1cm}f_t(x_{t-1:0}^s,u_{t-1:0}^s) \end{bmatrix}\hspace{-0.1cm},\hspace{-0.1cm}\label{eq:Neur-SLS_controller}\\
  &~&& t = 0,1,\ldots\,,~~~\xi_{-1} = 0\,.\label{eq:Neur_SLS_end}
  \end{alignat}
  
In the above problem, $\mathbf{K}'$ is an $\ell_2$-IS stabilizing base controller, and $\{w_{T:0}^s\}_{s=1}^S$ is a given training set of disturbances. The cost function \eqref{eq:loss} is defined as the empiric average of the loss evaluated over the training set, and the system dynamics are imposed through \eqref{eq:dynamics_s} for every $\mathbf{w}^s$. The relationships \eqref{eq:Neur-SLS_controller}-\eqref{eq:Neur_SLS_end} define a control sequence $\mathbf{u}^s = \mathbf{K}'(\mathbf{x}^s)+\Emme[\theta](\mathbf{x}^s-\mathbf{F}(\mathbf{x}^s,\mathbf{u}^s))$, where $\Emme[\theta] \in \mathcal{L}_2$ for every $\theta$. As a result, each value of $\theta \in \mathbb{R}^d$ yields closed-loop maps $(\Psix,\Psiu) \in \mathcal{CL}_2$.


We remark that the class of  all $\ell_2$-stable REN operators may be significantly more restrictive than the class of all operators in $\mathcal{L}_2$. Hence, learning over all $\ell_2$-stable REN operators may not allow reaping the benefits of the completeness result of  point $2)$ of Theorem~\ref{th:result}. This is why in the  REN-based Neur-SLS implementation \eqref{eq:loss}-\eqref{eq:Neur_SLS_end} we allow $\mathbf{K}'$ being $\ell_2$-IS stabilizing, but not necessarily strongly $\ell_2$-IS stabilizing. Based on the above reasoning, an important takeaway of Theorem~\ref{th:result} is that developing finite-dimensional approximations of $\mathcal{L}_p$ that are \emph{as large as possible} is a crucial endeavor — and a worthy one — towards globally optimal solutions of NOC. 



%

\section{Numerical experiments}
\label{sec:numerical}
\allowdisplaybreaks

In this section, we illustrate the use of Neur-SLS to tackle NOC problems. Further,  we validate the closed-loop stability guarantees during and after training. 
We implement the Neur-SLS \eqref{eq:loss}-\eqref{eq:Neur_SLS_end} using PyTorch and train the resulting DNN with stochastic gradient descent. 
The code to reproduce the examples and the 
implementation details are available at \url{https://github.com/DecodEPFL/neurSLS.git}.

We consider point-mass vehicles with position $p_t \in \mathbb{R}^2$ and velocity $q_t \in \mathbb{R}^2$ subject to nonlinear drag forces (e.g., air or water resistance). The discrete-time model for each vehicle of mass $m\in \mathbb{R}$ is 
\begin{equation}
\label{eq:mechanical_system}
    \begin{bmatrix}p_{t}\\q_{t}\end{bmatrix}=  \begin{bmatrix}p_{t-1}\\q_{t-1}\end{bmatrix} + T_s\begin{bmatrix}q_{t-1}\\m^{-1}\left(-C(q_{t-1})q_{t-1}+F_{t-1}\right)\end{bmatrix}\,,
\end{equation}
where $F_t \in \mathbb{R}^2$ denotes the force control input, $T_s>0$ is the sampling time and $C:\mathbb{R}^2\rightarrow \mathbb{R}$ is a positive \emph{drag function}. Typical scenarios include $C(q) = b|q|_2$ and $C(q) = b$ for some $b>0$. 
By comparing \eqref{eq:mechanical_system} with \eqref{eq:system}, we remark that the disturbance sequence is given by $((p_0,q_0,0,0),(0,0,0,0),\ldots) \in \ell_2^4$. Letting $\overline{p}\in \mathbb{R}^2$ and $\overline{q} = 0_2$ be a target position to be reached with zero velocity, we consider a  base controller $u_t = K'(\overline{p}-p_{t})$ with $K' = \operatorname{diag}(k_1',k_2')$ and $k_1',k_2'>0$. Note that the base controller is  strongly $\ell_2$-IS stabilizing if $C(q) = b$, and strongly $\ell_\infty$-IS stabilizing if $C(q)$ includes the nonlinearity $b|q|_2$.

We consider $N \in \mathbb{N}$ vehicles \eqref{eq:mechanical_system} with state $x_t \in \mathbb{R}^{4N}$ and input $u_t \in \mathbb{R}^{2N}$. 
We select the stage loss  
in \eqref{eq:loss} as
\begin{equation}\label{eq:loss_CA}
  l(x_t,u_t)= l_{traj}(x_t,u_t) + l_{ca}(x_t)+ l_{obs}(x_t)\,,
\end{equation}
where $l_{traj}(x_t,u_t) = \begin{bmatrix}x_t^\mathsf{T}&u_t^\mathsf{T}\end{bmatrix}Q\begin{bmatrix}x_t^\mathsf{T}&u_t^\mathsf{T}\end{bmatrix}^\mathsf{T}$ with $Q\succeq 0$ penalizes the distance of agents from their targets and the control energy, $l_{ca}(x_t)$ and $l_{obs}(x_t)$ penalize collisions between agents and with obstacles, respectively.

\subsection{Results}

The scenario \texttt{mountains} in Figure~\ref{fig:corridor} involves two agents 
whose goal is to coordinately pass through a narrow valley.
Both agents start from a randomly sampled initial position marked with ``$\circ$''. 
The scenario \texttt{swapping} in Figure~\ref{fig:robots_12} considers twelve agents switching their positions, while avoiding all collisions.\footnote{The \texttt{mountains} and \texttt{swapping} benchmarks are motivated by the examples in \cite{onken2021neural, furieri2021distributed}.} Animations  are available in our  \href{https://github.com/DecodEPFL/neurSLS.git}{Github repository.} We train Neur-SLS control policies to optimize the performance \eqref{eq:loss}-\eqref{eq:loss_CA} over a horizon of $5$ seconds with sampling time $T_s=0.05\si{\second}$, resulting in $T = 100$ time-steps. We consider  a linear drag force $C(q)q = bq$ for \texttt{swapping} and a non-linear drag force $C(q)q=b_1q+b_2|q|_2q$ for \texttt{mountains}, with suitable $b,b_1,b_2>0$. The trajectories after training are reported in Figures \ref{fig:corridor} and \ref{fig:robots_12}. The trained Neur-SLS control policies avoid collisions and achieve optimized trajectories thanks to minimizing \eqref{eq:loss_CA}.
Despite training for a finite-horizon, Theorem~\ref{th:result} guarantees closed-loop stability (i.e., $\ell_2$-stability for \texttt{swapping} and $\ell_\infty$-stability for \texttt{mountains}) around the target positions for $t\rightarrow \infty$. For validation, 
in Figure \ref{fig:loss} we consider  \texttt{mountains}  and report the cumulative loss $\sum_{k=0}^{t}l(x_k^s,u_k^s)$ at $t=0,1,\dots,500$ for $10$ randomly sampled initial conditions $\{w_0^s\}_{s=1}^{10}$  before and after the training. The cumulative loss converges to a constant value before and after training. 
Hence, the vehicles reach the target formation with vanishing control inputs as $t\rightarrow \infty$ irrespective of the chosen DNN parameters. This fact validates the built-in $\ell_\infty$-stability of the closed-loop maps.








\begin{figure}
  \centering
  \begin{minipage}{0.32\linewidth}
      \includegraphics[width=\linewidth]{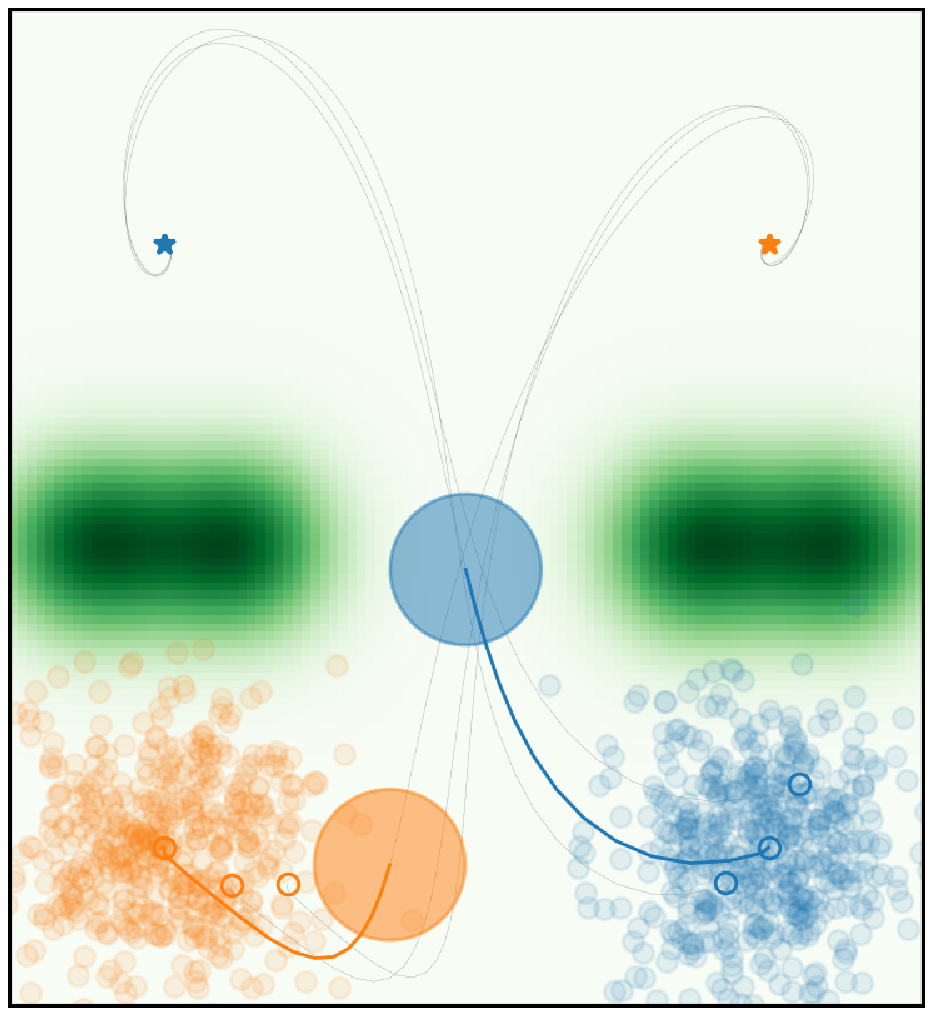}
  \end{minipage}%
  \begin{minipage}{0.32\linewidth}
      \includegraphics[width=\linewidth]{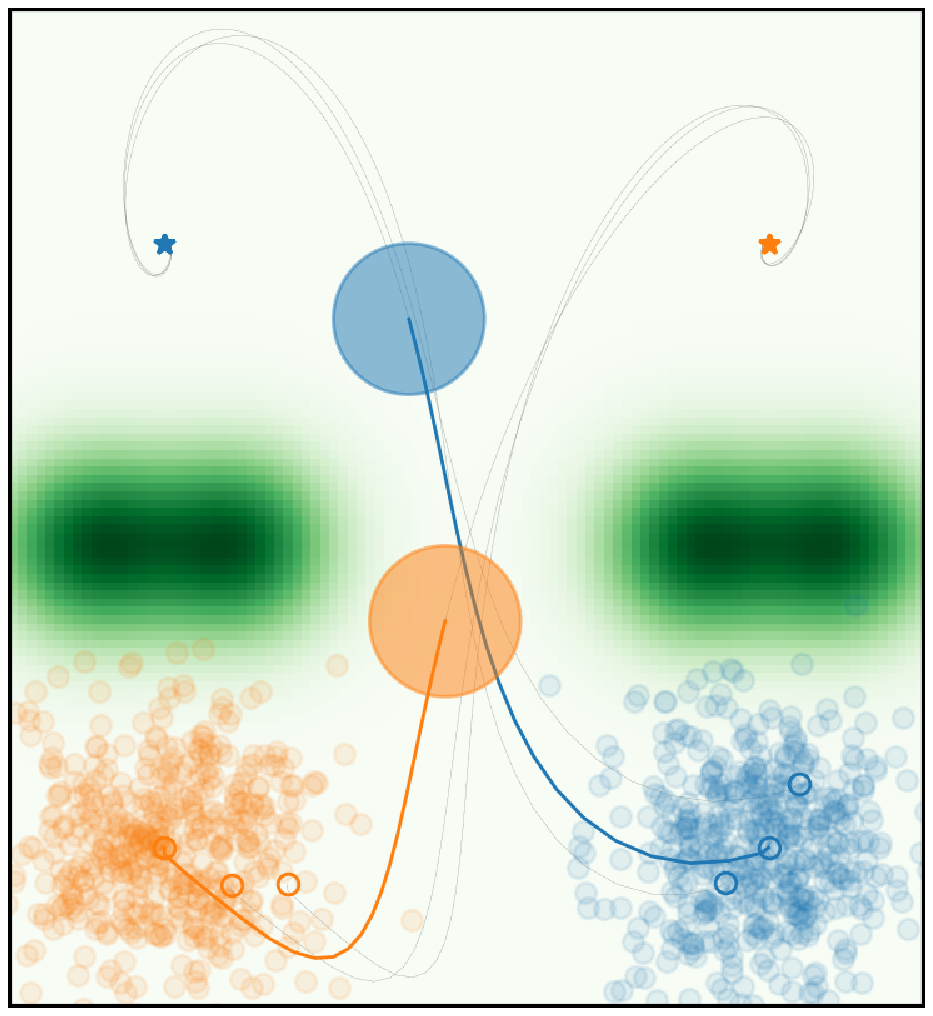}
  \end{minipage}%
  \begin{minipage}{0.32\linewidth}
      \includegraphics[width=\linewidth]{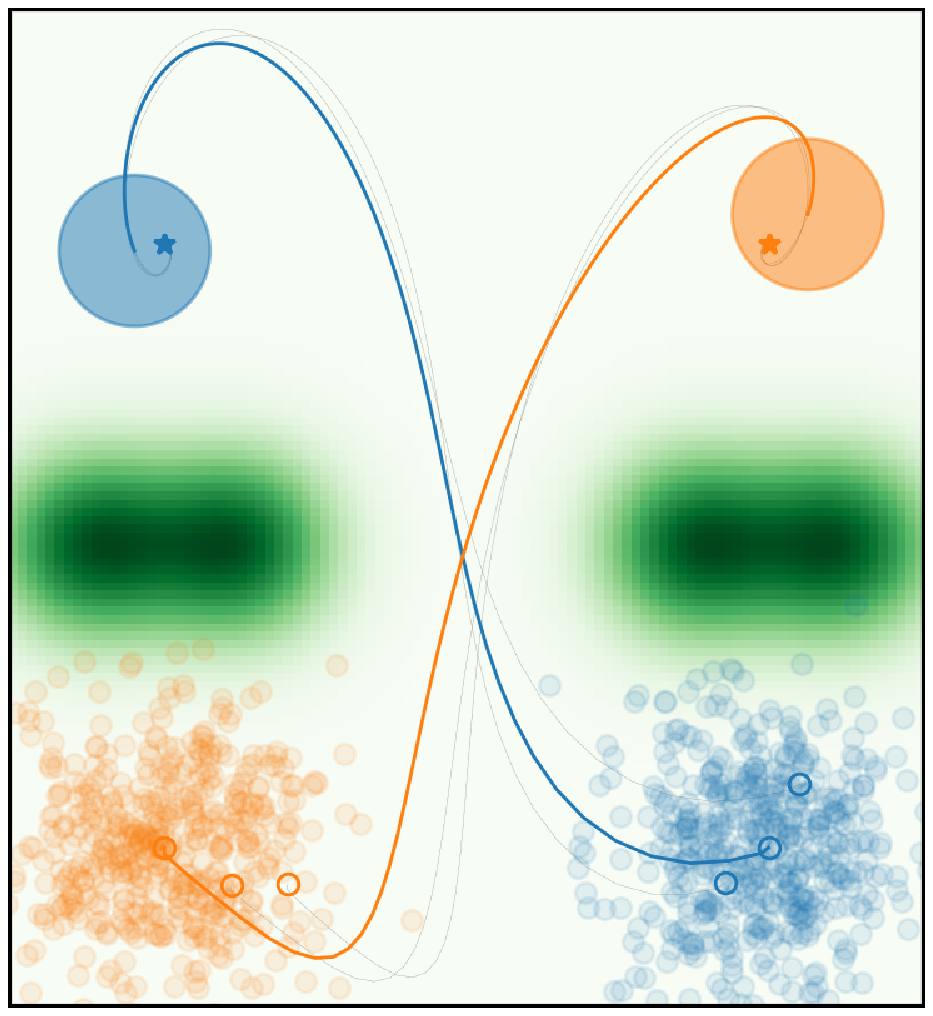}
  \end{minipage}
  \caption{\texttt{Mountains} --- Closed-loop trajectories after training over 500 randomly sampled initial conditions marked with $\circ$. Snapshots taken at times $\tau_1=0.7$, $\tau_2=1.1$ and $\tau_3=5.0$ seconds. Colored (gray) lines show the trajectories in $[0,\tau_i]$ ($[\tau_i,\infty)$). Colored balls (and their radius) represent the agents (and their size for collision avoidance).} 
  \label{fig:corridor}
\end{figure}

\begin{figure}
  \centering
    \begin{minipage}{0.30\linewidth}
      \includegraphics[width=\linewidth]{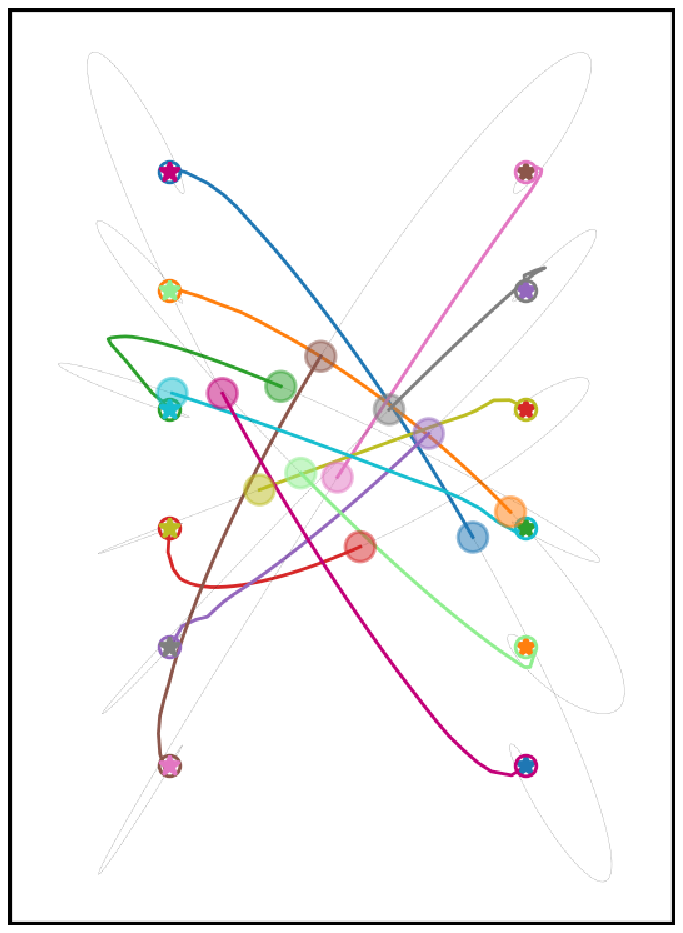}
  \end{minipage}%
  \begin{minipage}{0.30\linewidth}
      \includegraphics[width=\linewidth]{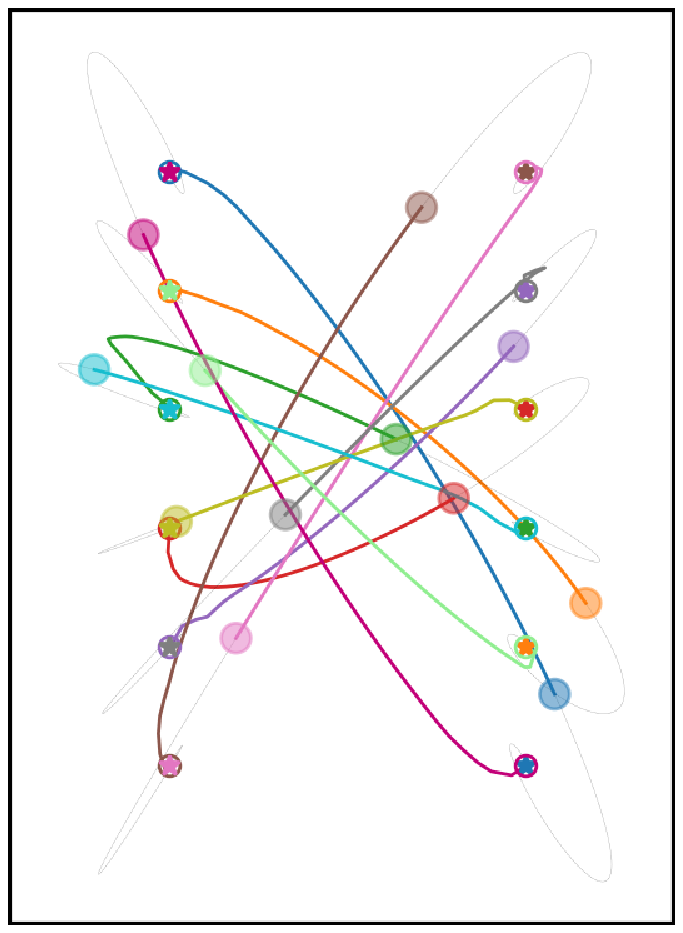}
  \end{minipage}%
  \begin{minipage}{0.30\linewidth}
      \includegraphics[width=\linewidth]{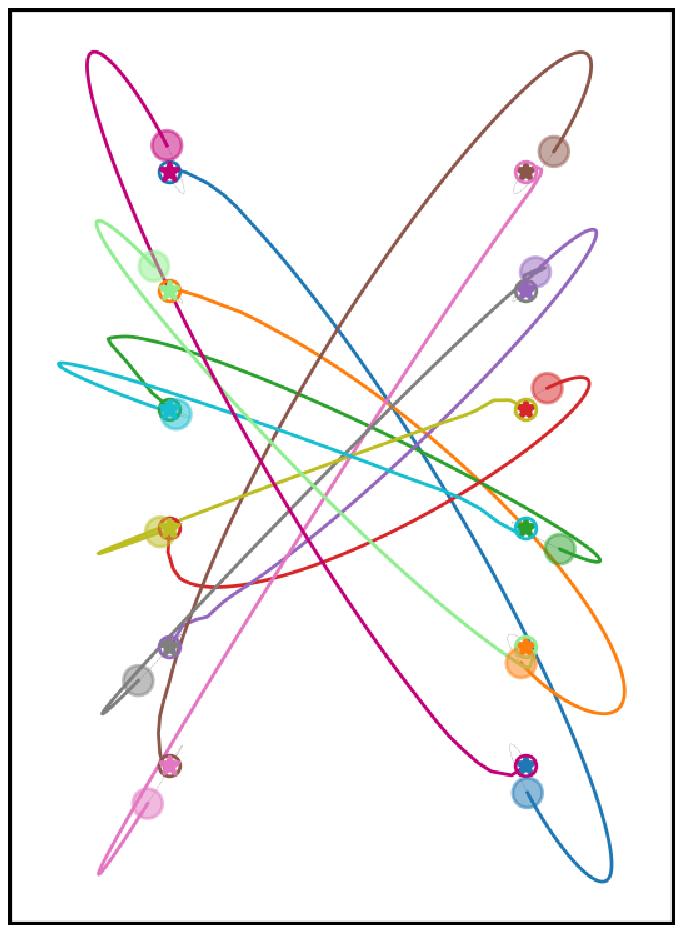}
  \end{minipage}
  \caption{\texttt{Swapping} --- Closed-loop trajectories after training. Snapshots taken at times $\tau_1=1.25$, $\tau_2=1.75$ and $\tau_3=5.0$ seconds. Colored (gray) lines show the trajectories in $[0,\tau_i]$ ($[\tau_i,\infty)$). Colored balls (and their radius) represent the agents (and their size for collision avoidance).}
  \label{fig:robots_12}
\end{figure}

\begin{figure}
  \centering
  \includegraphics[width=0.87\linewidth]{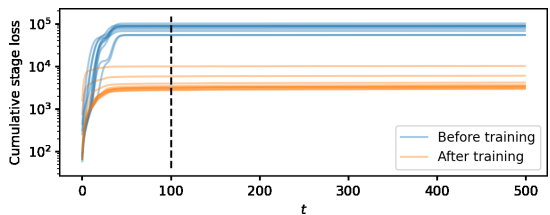}
  \caption{Cumulative stage loss in \texttt{mountains} before and after training. The black dashed line indicates the training horizon $T=100$.}
  \label{fig:loss}
\end{figure}


\smallskip

\section{Conclusions}
As we move towards designing deep nonlinear policies for general optimal control problems, it is crucial to guarantee closed-loop stability during and after training. We embrace a system-level perspective and propose a two-step procedure to parametrize all and only the stabilizing controllers in terms of one stable Youla-like operator. This results in a Neur-SLS optimization problem that can be tackled by training a DNN with unconstrained gradient descent. The proposed system-level perspective may open up several venues for future research, including output-feedback and applications to constrained, distributed and data-driven nonlinear control.

\section*{Acknowledgments}
We thank Dr. Alessandro Bosso for helpful discussions.
\bibliographystyle{IEEEtran}
\bibliography{references.bib}

\end{document}